\documentclass[journal]{IEEEtran}

\usepackage[T1]{fontenc}
\usepackage[utf8]{inputenc}
\usepackage{amsmath,amsthm,mathtools,amssymb}
\mathtoolsset{showonlyrefs=true}

\usepackage{graphicx}

\usepackage{dsfont}
\usepackage{todonotes}

\usepackage{microtype}

\usepackage{array}
\usepackage{ragged2e}
\newcolumntype{L}[1]{>{\RaggedRight\arraybackslash}p{#1}}

\usepackage{url}

\usepackage[breaklinks]{hyperref}

\usepackage{orcidlink}

\newtheorem{corollary}{Corollary}

\newtheorem{theorem}{Theorem}

\newtheorem{assumptions}{Assumptions}

\usepackage{tikz}
\usetikzlibrary{calc,positioning,fit,decorations.pathreplacing}

\newcommand{\matr}[1]{{\boldsymbol{#1}}}
\renewcommand{\vec}[1]{{\boldsymbol{#1}}}
\newcommand{\PP}{\mathcal{P}}
\newcommand{\QQ}{\mathcal{Q}}
\newcommand{\eye}{\mathds{1}}

\begin{document}

\title{A graph-theoretic view on the small-signal stability of inverter-based power grids}

\author{%
Iva Ba\v{c}i\'{c}\orcidlink{0000-0003-2987-5065},
Jakob Niehues\orcidlink{0000-0003-0140-6910},
Philipp C.~Böttcher\orcidlink{0000-0002-3240-0442},
Cai Dieball\orcidlink{0000-0002-0011-2358},
Leonardo Rydin Gorjão\orcidlink{0000-0001-5513-0580}, 
Andrea Benigni\orcidlink{0000-0002-2475-7003}, \textit{Senior Member, IEEE}, 
Frank Hellmann\orcidlink{0000-0001-5635-4949},
Dirk Witthaut\orcidlink{0000-0002-3623-5341}
\thanks{
I.~Bačić, P.~C.~Böttcher, A.~Benigni and D.~Witthaut are with the Forschungszentrum J\"ulich, Institute of Climate and Energy Systems: Energy Systems Engineering (ICE-1), D-52428 J\"ulich, Germany.
I.~Ba\v{c}i\'{c} is with the Institute of Physics Belgrade, University of Belgrade, Serbia.
J.~Niehues and F.~Hellmann are with the Potsdam Institute for Climate Impact Research (PIK), D-14412 Potsdam, Germany.
J.~Niehues is with Technische Universität Berlin, D-10623 Berlin, Germany.
C.~Dieball and D.~Witthaut are with the Institute for Theoretical Physics, University of Cologne, D-50937 K\"oln, Germany. 
L.~Rydin Gorjão is with the Department of Environmental Sciences, Open University of the Netherlands, Heerlen, The Netherlands \& Institute of Physics, Norwegian University of Life Sciences, Ås, Norway.
A.~Benigni is with RWTH Aachen University, Germany.
}
}

\maketitle

\begin{abstract}
    Dynamic grid stability is traditionally ensured with synchronous generators.
    Modern grids rely substantially more on inverter-based resources, which require grid-forming control to guarantee adequate system-wide synchronization and stability.
    Small-signal stability has granted various centralized and decentralized stability certificates --- but these have primarily been limited to sufficient criteria only.
    In this work, we construct a necessary and sufficient small-signal stability criterion for lossless inverter-based power grids with arbitrary topology.
    We show that asymptotic stability is equivalent to the positive definiteness of a single matrix that combines network topology, operating point, and effective droop gains.
    We derive graph-theoretic stability criteria based on an augmented cone graph and show that the contribution of graph cycles, though not itself small, barely affects the certified stability boundary, as illustrated for three IEEE test cases.
    The resulting framework yields decentralized stability criteria, quantifies the conservatism introduced by decentralization, and may support the development of future grid codes.
\end{abstract}

\section{Introduction}

The ongoing energy transition poses difficult questions for the stable operation and design of power grids~\cite{witthautCollective2022}.
One central challenge is to ensure system-wide frequency synchronization and dynamic stability~\cite{ENTSOE2020}. Traditionally, this was the task of synchronous generators (SGs).
As they are progressively replaced by inverter-based resources (IBRs), grid-forming control is essential for maintaining system stability.
Accordingly, grid codes are currently being revised to define the required grid-forming capabilities and to ensure that large populations of IBRs can collectively replace the stabilizing role of synchronous generators.

While SGs are connected to the closely monitored and actively managed transmission grid, IBRs are predominantly connected at the distribution level.
Complete and continuously updated knowledge of the network state, topology, and device parameters can therefore not be assumed for system-wide stability assessment.
These challenges call for a deeper understanding of small-signal stability and for decentralized stability certificates that rely only on locally available information and can be translated into grid-code requirements.

The field of small-signal stability of multi-machine power grids has a long tradition, going back to the seminal work of Bergen and Hill~\cite{bergenStructurePreserving1981}.
Since then, a plethora of results from power engineering~\cite{Machowski2020}, control theory~\cite{dorflerSynchronization2013, schiffer2014conditions}, and nonlinear dynamics~\cite{witthautCollective2022} have provided a broad range of analytical tools for studying the small-signal stability of power systems.

Here, we derive a necessary and sufficient criterion for the small-signal stability of synchronization in lossless inverter-based power grids.
Going beyond existing sufficient criteria, we provide sharp necessary criteria without assuming a restrictive device model or a specific topology.
We show that stability is equivalent to positive definiteness of a single matrix that combines network topology, operating point, and inverter dynamics.
This criterion admits a graph-theoretic interpretation, which clarifies how nodal properties and line loadings jointly determine system stability.
This view naturally leads to decentralized sufficient criteria involving only local node and edge properties, generalizing previous results in the literature.
Moreover, it allows us to quantify the conservatism of such local criteria by identifying the stabilizing network interactions they neglect.

The remainder of this paper is organized as follows.
Section~\ref{sec:literature} reviews the existing literature on decentralized small-signal stability criteria for inverter-based power grids.
Section~\ref{sec:problem} introduces the system model and formulates the small-signal stability problem.
In Section~\ref{sec:stability-fundamental}, we derive the fundamental necessary and sufficient stability criterion and introduce a graph-theoretic perspective.
Section~\ref{sec:local} develops decentralized local stability certificates.
Finally, Section~\ref{sec:examples} illustrates the theoretical results using representative examples.

\section{Literature Review}\label{sec:literature}
The literature on small-signal stability of inverter-based grids is dominated by \emph{sufficient} criteria.
Results that are simultaneously necessary and sufficient, and a clean account of \emph{when} a sufficient criterion is tight, are both rare and rely on strong modeling assumptions.

\begin{table*}[t]
\centering
\caption{Necessary-and-sufficient stability criteria across the literature. Entries describe the scope each work adopts.}
\label{tab:nec-suf}
\renewcommand{\arraystretch}{1.35}
\footnotesize
\begin{tabular}{@{}L{3.0cm}L{3.1cm}L{1.8cm}L{2.4cm}L{2.2cm}L{3.5cm}@{}}
\hline\hline
\textbf{Reference} & \textbf{Device models} & \textbf{Device heterogeneity} & \textbf{Network model} & \textbf{$v$--$\omega$ coupling} & \textbf{Criterion type} \\
\hline
Groß~\cite{grossCompensatingNetworkDynamics2022b} & $p$-$\omega$ droop & scaled homogeneous & uniform $R/X$, $v^\circ$, $\theta^\circ$; dynamic & no & nec. (instability); suf. (stability); gap open\\
Dörfler et al.~\cite{dorflerSynchronization2013}; Simpson-Porco et al.~\cite{simpson-porcoSynchronizationPowerSharing2013} & $p$-$\omega$ droop & heterogeneous & lossless, static & no & nec. \& suf. on acyclic/selected graphs\\
Gorbunov et al.~\cite{gorbunovIdentificationStabilityRegions2022} & $p$-$\omega$ \& $q$-$v$ droop & homogeneous droop ratio & lossy, dynamic & network only & nec. \& suf. for uniform $R/X$ \& droop ratio\\
Nishino et al.~\cite{nishinoSmallSignalStabilityCondition2025} & electromagnetic 4th/2nd/1st-order& heterogeneous & lossless, static & network \& devices (4th-order: $p$-$v$) & nec. \& suf. (dev.-indep.) \\
\textbf{This work} & arbitrary with $q$-$v$ droop, passive through loop shift & heterogeneous & lossless, static & network \& devices & nec. \& suf. (dev.-indep. except effective droop gain)\\
\hline\hline
\multicolumn{6}{l}{{\scriptsize (nec.: necessary; suf.: sufficient; dev.-indep.: device-independent -- no device parameters in the criteria)}}
\end{tabular}
\vspace*{-0.5cm}
\end{table*}

A large body of work develops increasingly \emph{decentralized} sufficient criteria, ranging from centralized port-Hamiltonian analyses~\cite{schiffer2014conditions} through $\mathcal{H}_\infty$ criteria~\cite{patesDecentralisedRobustInverterbased2017}, and secant/Popov~\cite{monshizadehSecantPopovlikeConditions2019},
to mixed gain-phase~\cite{huangGainPhaseDecentralized2024}, graphical separation of projected Davis--Wielandt shells~\cite{huangGeometricDecentralizedStability2025a,zhangPhantomDavisWielandtShell2026}, and loop-shifted passivity~\cite{haberleDecentralizedParametricStability2026}, including the voltage-droop criterion of~\cite{niehuesSmallSignal2026}.
These certificates rely on local device data and, at most, aggregated network scalars, deliberately discarding graph structure.
This is what makes them scalable, but is also the structural source of their conservativeness~\cite{baron-pradaImpactOperatingPoints2026}, which is known to increase with network stress~\cite{niehuesSmallSignal2026,wangDecentralizedStabilityCertificates}.

In contrast, sharp necessary criteria for stability have required strong simplifying assumptions.
Classical closed-form necessary-and-sufficient criteria are only available for highly simplified models on restricted graphs: the synchronization criteria of Dörfler et al.~\cite{dorflerSynchronization2013} and Simpson-Porco et al.~\cite{simpson-porcoSynchronizationPowerSharing2013} are exact on acyclic and other selected topologies and under a Kuramoto reduction.
Voltage-magnitude dynamics enter only as a robustness margin, not as coupled dynamics.
Nishino et al.~\cite{nishinoSmallSignalStabilityCondition2025} obtain a necessary and sufficient criterion in the network susceptance matrix and the stationary power flow that is independent of the device dynamics, which is a consequence of electromagnetic modeling of SGs and grid-forming inverters as a singular perturbation without voltage-control dynamics.
The generalized short-circuit ratio of Xin et al.~\cite{xinGeneralizedShortCircuit2018} is exact by decomposing the system node by node, which rests on node homogeneity and a no-power network approximation, while results of Gorbunov et al.~\cite{gorbunovIdentificationStabilityRegions2022} are exact under a uniform network $R/X$ ratio and a uniform droop ratio across inverters.
The comparison of stability and instability criteria by Groß~\cite{grossCompensatingNetworkDynamics2022b} shows that the necessary-sufficient gap is, in general, open.
Table~\ref{tab:nec-suf} compares the most relevant literature, ordered according to the modeling assumptions that determine the exactness of the resulting stability criteria.

Our analysis avoids typical simplifying assumptions and focuses on systems with arbitrary network topology and almost arbitrary device dynamics.
The devices can be heterogeneous and implement arbitrary coupling between active and reactive power as inputs and frequency and voltage magnitude as outputs.
The only structural assumption on the device dynamics is the presence of reactive-power-voltage droop that enables loop-shifted passivity.
We adopt a lossless static network, consistent with the aforementioned work and with the inductive transmission regime, and regard the uniform $R/X$~\cite{niehuesSmallSignal2026} and lossy~\cite{balestraMultistabilityLossyPower2019} settings as the natural route to relaxing it.
The cross-coupling of frequency and voltage-magnitude dynamics over the network is retained explicitly.
Then, the existing decentralized sufficient criteria in~\cite{niehuesSmallSignal2026} reappear as more conservative special cases of our exact criterion.
This also explains why a separable Lyapunov certificate, which underlies such decentralized criteria, is known to be necessary as well as sufficient only on trees~\cite{bermanMatrixDiagonalStability1983}, single cycles~\cite{arcakDiagonalStabilityClass2006}, and cacti~\cite{arcakDiagonalStabilityCactus2011}.
On general meshed graphs, it is merely sufficient, in line with the stabilizing role of cycles long known in oscillator networks~\cite{witthautBraesssParadoxOscillator2012,bronskiGraphHomologyStability2016,tylooRobustnessSynchronyComplex2018,witthautCollective2022}.

\section{Problem Statement}\label{sec:problem}

\subsection{Variables and network model}

We consider the small-signal, i.e. linear, stability of power grid dynamics in power-phase coordinates around a fixed operating point.
The fundamental variables that characterize the power grid's state are the voltage amplitude $v_i$ and phase $\theta_i$ at bus $i$ and the active and reactive power $p_i$ and $q_i$ injected by node $i$ into the grid.
An operating point is characterized by uniformly oscillating voltages $\theta^\mathrm{op}_i(t) = \theta^\circ_i + \omega^\circ t$ with constant amplitudes $v^\circ$, active power $p_i^\circ$, and reactive power $q_i^\circ$.

It is often convenient to work with the frequency $\omega_i = \dot \theta_i$, the logarithm of the voltage amplitude $u_i = \ln(v_i)$ and its derivative, the relative amplitude velocity $\dot u = \frac{\mathrm{d}}{\mathrm{d}t}\ln(v) = \frac{\dot v}{v}$.
We refer to~\cite{haberleDecentralizedParametricStability2026} for inferring voltage stability from the stability of its velocity.
Linearizing around $v^\circ$ we have $\delta v \approx v^{\circ} \delta u$.
With this, $u + j \theta$ is the complex phase~\cite{niehuesComplexPhase2025} and its derivative $j \omega + \dot u$ is the complex frequency~\cite{milanoComplexFrequency2021}.

Power-phase coordinates have several advantages: For many common IBR control schemes, the linearization in power-phase coordinates is only weakly dependent on the set points and grid state around which we linearize~\cite{koglerNormalForm2022, buttnerComplexPhase2025}.
Further, the linearization does not depend on arbitrary phases and stays valid during phase drifts \cite{niehuesComplexPhase2025}.
Finally, the small-signal response of the network in power-phase coordinates captures the effect of line loading on small-signal stability explicitly, and the small-signal response of the nodes, which relates active power and reactive power to frequency and voltage amplitude, closely aligns with the performance requirements in power grids~\cite{haberle2026next}.

Power flows are described  by the lossless load-flow equations,
\begin{align}
    p_i &= \sum_{j=1}^n v_i v_j  B_{ij} \sin(\theta_i - \theta_j) 
    + p_i^{\rm load},\\
    q_i &= - \sum_{j=1}^n v_i v_j B_{ij} \cos(\theta_i - \theta_j) 
    + q_i^{\rm load},
\end{align}
where we assume that the network is connected and $B_{ij} = B_{ji} \ge 0 $ for $i\neq j$ and $B_{ii} = B_{ii}^{\rm shunt} - \sum_{j \neq i } B_{ij}$. 
Throughout this manuscript, we assume that shunt elements are such that $B_{ii} < 0$ for all $i$. The network equations can be assumed static, provided that line dynamics are sufficiently fast compared to the device dynamics~\cite{Machowski2020} or compensated by converter control \cite{grossCompensatingNetworkDynamics2022b}.

Here we assume that every node $i$ is associated with a device that injects or extracts active and reactive power and has some grid-forming capability. 
Furthermore, we allow for static loads connected to the same nodes with fixed power consumption $p_i^{\rm load}$ and $q_i^{\rm load}$. 
We can further allow for nodes with instantaneous loads and no active devices, as in the classical model of power system dynamics~\cite{Machowski2020}.
These nodes can be eliminated via a Kron reduction~\cite{dorflerKron2013}, leading to an effective network with effective parameters $B_{ij}$.

At a given operating point, the small-signal response of the nodal power injection to small changes in phase $\delta \theta_i \coloneqq \theta_i - \theta^\circ_i$ and amplitude $\delta u_i \coloneqq u_i - u_i^\circ$ is given by
\begin{equation}
    \begin{pmatrix}
        \vec p - \vec p^\circ \\ \vec q - \vec q^\circ
    \end{pmatrix} = 
    \begin{pmatrix}
        \matr{M}^{p \theta} & \matr{M}^{p v} \\
        \matr{M}^{q \theta} & \matr{M}^{q v} 
    \end{pmatrix}
    \begin{pmatrix}
        \vec \delta \vec \theta \\ \vec \delta \vec u
    \end{pmatrix} \eqqcolon \matr M     \begin{pmatrix}
        \vec \delta \vec \theta \\ \vec \delta \vec u
    \end{pmatrix}.
\end{equation}
The respective matrices are given by
\begin{align}
    {M}^{p \theta}_{ij} &= 
    \begin{cases}
        - B_{ij} v_i^\circ v_j^\circ \cos(\theta_i^\circ - \theta_j^\circ)
        & \mbox{for } i \neq j \\
        \sum_{k \neq i}
         B_{ik} v_i^\circ v_k^\circ \cos(\theta_i^\circ - \theta_k^\circ)
         & \mbox{for } i = j, %\\
    \end{cases} \\
    {M}^{p v}_{ij} &= 
    \begin{cases}
         B_{ij} v_i^\circ v_j^\circ \sin(\theta_i^\circ - \theta_j^\circ)
         & \mbox{for } i \neq j \\
        \sum_{k \neq i} B_{ik} v_i^\circ v_k^\circ \sin(\theta_i^\circ - \theta_k^\circ)
        & \mbox{for } i = j, \\
    \end{cases} \\
    {M}^{q \theta}_{ij} &= M^{p v}_{ji}, \\
    {M}^{q v}_{ij} &= 
    \begin{cases}
        - B_{ij} v_i^\circ v_j^\circ \cos(\theta_i^\circ - \theta_j^\circ)
        &  i \neq j \\
        - 2 B_{ii} (v_i^\circ)^2 - \sum_{k \neq i} B_{ik} v_i^\circ v_k^\circ \cos(\theta_i^\circ - \theta_k^\circ)
         &  i = j. 
    \end{cases}    
\end{align}
As noted above, we will also work with frequency and amplitude velocity.
With $\vec \delta \vec p \coloneqq \vec p^\circ - \vec p$ and $\vec \delta \vec q \coloneqq \vec q^\circ - \vec q$, the small-signal response in the Laplace domain is
\begin{equation}
    -\begin{pmatrix}
        \vec \delta \vec p \\ \vec \delta \vec q
    \end{pmatrix} =  \frac{1}{s} \matr M
    \begin{pmatrix}
        \vec \delta \vec \omega \\ \vec \delta \dot{\vec u}
    \end{pmatrix}
    \eqqcolon  \matr N(s) \begin{pmatrix}
        \vec \delta \vec \omega \\ \vec \delta \dot{\vec u}
    \end{pmatrix}.
\end{equation}
In power-phase coordinates, the network response is not passive.
Eigenvalues with a negative real part typically occur in the voltage sector $\matr{M}^{q v}$.
As we will discuss below, shifting part of the nodal dynamics to the network side can make the network passive~\cite{haberleDecentralizedParametricStability2026,niehuesSmallSignal2026}.

\subsection{Device model}
As a counterpart to the network equations, the devices encode how local frequency and amplitude velocities $(\omega_i,\dot u_i)$ react to power and reactive power at the node $(p_i,q_i)$.
We assume that device dynamics can be described by a differential equation on internal states $\vec x_i$, $\theta_i$ and $u_i$ with inputs $p_i$ and $q_i$.
Due to symmetry there is no explicit dependence on $\theta_i$ \cite{koglerNormalForm2022, buttnerComplexPhase2025}, and its equations are of the form
\begin{align}
    \dot{\vec x}_i &= f_i(\vec x_i, u_i, p_i, q_i), \\    
    \begin{pmatrix}
        \dot{\theta}_i \\ \dot u_i
    \end{pmatrix} &= g_i(\vec x_i, u_i, p_i, q_i).
\end{align}
For our analysis, we make the structural assumption introduced in~\cite{niehuesSmallSignal2026}: the nodes react proportionally to reactive power and to voltage disturbances; in other words, they implement a type of $q$-$v$ droop with effective droop constant $k^q_i$.
The equations then take the explicit form
\begin{align}
    \dot{\vec x}_i &= \tilde f_i(\vec x_i, p_i, k_i^q q_i + u_i), \\   
    \begin{pmatrix}\dot{\theta}_i \\ \dot u_i\end{pmatrix} &= \tilde g_i(\vec x_i, p_i, k_i^q q_i + u_i).
\end{align}
Under this assumption, \cite{niehuesSmallSignal2026} proved sufficient, decentralized stability criteria.
As we will see below, this assumption allows shifting the negative voltage feedback to the network side.
This has a dual effect: It makes typical device models as well as the network response passive.

The small-signal response of frequency and voltage velocity to active and reactive power disturbances can then be separated as
\begin{align}
    \begin{pmatrix}
        \delta \omega_i \\ \delta \dot u_i
    \end{pmatrix} &= \matr D_i(s)
    \begin{pmatrix}
        p_i^\circ - p_i \\
        q_i^\circ - q_i 
    \end{pmatrix} \\
    &= \matr{ \tilde D}_i(s)
    \begin{pmatrix}
        p_i^\circ - p_i  \\
        q_i^\circ - q_i 
    \end{pmatrix} + \frac 1 {k^q_i} \matr{ \tilde D}_i(s)
    \begin{pmatrix}
        0 \\
        u^\circ_i - u_i
    \end{pmatrix}, \label{eq:general-D}
\end{align}
and we can treat $\delta \tilde q_i = \delta q_i - \frac1{k^q_i} \delta u_i$ as nodal input.
The response $\matr{\tilde D}_i$ can typically be read off directly from the equations.
To obtain the explicit relationship to $\matr D_i(s)$, define the $2{\times}2$ matrix $\matr{H}_i(s) = \operatorname{diag}( 0 , \frac 1 { k^q_i s} )$.
A direct calculation shows that
\begin{align}
     \matr D_i(s) &= \left(\matr \eye + \matr{\tilde D}_i(s) \matr{H}_i(s)\right)^{-1} \matr{\tilde D}_i(s), \\
     \matr {\tilde D}_i(s) &=  \matr{D}_i(s) \left(\matr \eye - \matr{H}_i(s) \matr{D}_i(s) \right)^{-1} . 
     \label{eq:D-tilde}
\end{align}

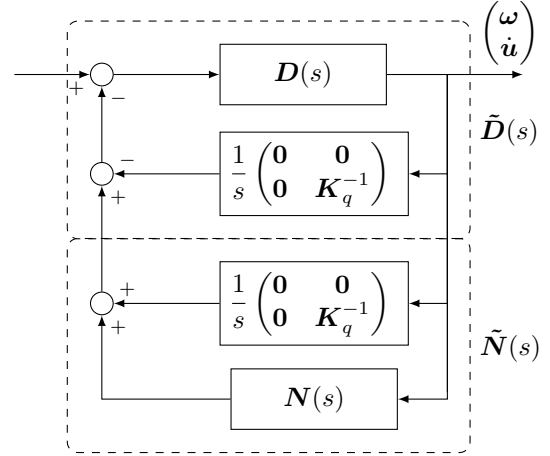
\begin{figure}[tb]
    \centering
    \begin{tikzpicture}[
      >=latex,
      block/.style={draw,rectangle,align=center,minimum height=8mm,minimum width=22mm},
      bigblock/.style={draw,rectangle,align=center,minimum height=12mm,minimum width=36mm},
      sum/.style={draw,circle,inner sep=0pt,minimum size=3mm},
      node distance=6mm
    ]
    \coordinate (in) at (0,0);
    \node[sum, right=10mm of in] (s1) {};
    \node[block, right=9mm of s1] (D) {\(\vec{D}(s)\)};
    \coordinate (tap) at ($(D.east)+(12mm,0)$);
    \coordinate (out) at ($(tap)+(12mm,0)$);
    \node[sum, color=white] (faketap) at (tap)  {};
    \draw[->] (in) -- (s1) node[above, pos=0.15] {\(\begin{pmatrix} \vec \delta \vec p_\mathrm{d} \\ \vec \delta \vec q_\mathrm{d} \end{pmatrix}\)};
    \draw[->] (s1) -- (D);
    \draw[->] (D.east) -- (out)
      node[pos=0.9,above] {\(\begin{pmatrix}\vec \delta \vec \omega \\ \vec\delta \dot{\vec u}\end{pmatrix}\)};
    \node[font=\scriptsize] at ($(s1.west)+(-1.0mm,-1.8mm)$) {$+$};
    \node[font=\scriptsize] at ($(s1.south)+(2.2mm,-1.6mm)$) {$-$};
    \node[sum, below=9mm of s1] (su) {};
    \node[block, right=8mm of su] (Ku) {%
      \(\dfrac{1}{s}
      \begin{pmatrix}
        \matr 0 & \matr 0 \\
        \matr 0 & \matr K_q^{-1}
      \end{pmatrix}\)
    };
    \coordinate (vrU) at ($(tap)+(0,-14mm)$);
    \draw[->] (tap) |- (Ku.east);
    \draw[->] (Ku.west) -- (su.east);
    \node[font=\scriptsize] at ($(su.east)+(1.8mm,1.8mm)$) {$-$};
    \draw[->] (su.north) -- (s1.south);
    \node[sum, below=15mm of su] (sl) {};
    \node[block, right=8mm of sl] (Kl) {%
      \(\dfrac{1}{s}
      \begin{pmatrix}
        \matr 0 & \matr 0 \\
        \matr 0 & \matr K_q^{-1}
      \end{pmatrix}\)
    };
    \coordinate (vrL) at ($(tap)+(0,-28mm)$);
    \draw[->] (tap) |- (Kl.east);
    \draw[->] (Kl.west) -- (sl.east);
    \node[font=\scriptsize] at ($(sl.east)+(1.8mm,1.8mm)$) {$+$};
    \node[font=\scriptsize] at ($(sl.south)+(2.2mm,-1.6mm)$) {$+$};
    \draw[->] (sl.north) -- (su.south);
    \node[font=\scriptsize] at ($(su.south)+(2.2mm,-1.6mm)$) {$+$};
    \node[block, below=2mm of Kl] (N) {\(\vec{N}(s)\)};
    \coordinate (vrN) at ($(tap)+(0,-42mm)$);
    \node[sum, right=10mm of N] (sN) {};
    \draw[->] (tap) -- (sN.north);
    \coordinate (lowerin) at ($(sN)+(12mm,0mm)$);
    \coordinate (lowerout) at ($(N.west)+(-22mm,0mm)$);
    \draw[->] (sN.west) -- (N.east);
    \draw[->] (lowerin) -- (sN.east) node[above, pos=0.1] {\(\begin{pmatrix}\vec{\delta\omega}_\mathrm{d}\\\vec{\delta}\dot{\vec{u}}_\mathrm{d}\end{pmatrix}\)};
    \draw[->] (N.west) -- (lowerout) node[above, pos=0.95] {\(\begin{pmatrix} \vec \delta \vec p\\ \vec \delta \vec q \end{pmatrix}\)};
    \node[font=\scriptsize] at ($(sN.north)+(2.2mm,1.6mm)$) {$+$};
    \node[font=\scriptsize] at ($(sN.south)+(2.2mm,-1.6mm)$) {$+$};
    \draw[->] ($(lowerout) + (11mm,0mm)$) -- (sl.south);
    \node[draw,dashed,rounded corners,inner sep=3mm,fit=(s1)(D)(su)(Ku)(vrU)(faketap)] (boxD) {};
    \node[right=0mm of boxD.east] {\(\vec{\tilde{D}}(s)\)};
    \node[draw,dashed,rounded corners,inner sep=3mm,fit=(sl)(Kl)(N)(vrN)(sN)] (boxN) {};
    \node[above right=6mm and 0mm of boxN.east] {\(\vec{\tilde{N}}(s)\)};
    \node[sum, right=10mm of in] (s1) {};
\end{tikzpicture}
\vspace*{-0.2cm}\caption{System model $\matr D(s)\#\matr N(s)$ with coordinate transformation and loop shift. Subscript $\mathrm{d}$ denotes an external disturbance.}\label{fig:system-full}
\end{figure}

We now return to the full system dynamics.
The change of coordinates from $\vec \delta \vec q$ to $\vec \delta \vec{\tilde q}$ induces a loop shift \cite{haberleDecentralizedParametricStability2026,niehuesSmallSignal2026,wangDecentralizedStabilityCertificates}.
Equation~\eqref{eq:general-D} shows that the loop-shifted devices are represented by $\matr{\tilde D}_i(s)$.
The network side is represented by
\begin{equation}
    \matr {\tilde N}(s) = \frac 1 s \matr \Xi \coloneqq \frac 1 s \left(\matr M + 
    \begin{pmatrix} 
        \matr 0  & \matr 0 \\
        \matr 0 & \matr{K}_q^{-1} 
    \end{pmatrix}\right). \label{eq:Xi}
\end{equation}
where $\matr K_q \coloneqq \operatorname{diag}(k^q_i)$ collects the effective droop constants.
Figure~\ref{fig:system-full} depicts the full interconnected system and the shifts.

\subsection{Special case: Diagonal node responses and droop laws}
For a diagonal node response, the separation above is always available.
The most important example in this model class is the droop-control law~\cite{schiffer2014conditions}, given as
\begin{align}
    \dot{\theta}_i&=\omega_i,\\
    \tau_i^p \dot\omega_i &= -\omega_i + \omega^{\circ} + \beta_i^p  ( p_i^\circ -p_i), \\
    \tau_i^q \dot v_i &= - v_i + v_i^\circ + \beta_i^q \left(q_i^{\circ}  -q_i \right).
\end{align}
Linearization does not change the first two equations; the last equation turns into
\begin{equation}
    \tau_i^q \delta \dot u_i = - \delta u_i 
    + \frac{\beta_i^q}{v_i^\circ} \left( q_i^{\circ} -q_i \right).
\end{equation}
From this, we can read off $k_i^q = \beta_i^q/v_i^\circ$ and 
\begin{equation}
    \matr{\tilde D}_{i,\mathrm{droop}}(s) = \begin{pmatrix}
        \frac{\beta_i^p}{1 + s \tau_i^p} & 0 \\ 0 & \frac{k_i^q}{\tau_i^q} 
    \end{pmatrix}.\label{eq:D-droop}
\end{equation}
While droop controllers are not passive in voltage-current coordinates, $\matr{\tilde D}_{i,\mathrm{droop}}(s)$ is.
Furthermore, the small-signal response of the droop law does not depend on $p_i^\circ$, $q_i^\circ$, and $\theta_i^\circ$. It depends on $v_i^\circ$, which only varies within tight bounds in practice.

\section{Fundamental Stability Results} \label{sec:stability-fundamental}
In this section, we derive fundamental results on the stability of inverter-based power grids in terms of the matrix $\matr \Xi$.
We assume that, for each device $i$, there exists a real number $k_i^q>0$ such that the loop-shifted transfer functions $\matr{\tilde D}_i(s)$ in~\eqref{eq:D-tilde} belong to $\mathcal{RH}_\infty$, i.e., they are real-rational, proper, and stable, and that they are output-strictly passive, as commonly assumed in the analysis of inverter-based power systems~\cite{haberleDecentralizedParametricStability2026,niehuesSmallSignal2026,wangDecentralizedStabilityCertificates,haberle2026next,peng2026positivedampingregiongraphic}. For our analysis, it is convenient to work with the inverse transfer functions
$\matr\Gamma_i(s)\coloneqq\matr{\tilde D}_i(s)^{-1}$ and formulate our assumptions through a uniform positive-real bound.

\begin{assumptions}[Inverter Transfer Functions]
\label{ass:passivity}
For every device $i$, the loop-shifted transfer function
$\matr{\tilde D}_i(s)\in\mathcal{RH}_\infty$ and there is a constant
$\gamma_i>0$ such that, on
$\overline{\mathbb C}_+\coloneqq\{s\in\mathbb C:\Re(s)\ge0\}$,
\begin{enumerate}
\item The inverse $\matr\Gamma_i(s)=\matr{\tilde D}_i(s)^{-1}$ exists and has no poles, and
\item $\matr\Gamma_i(s)+\matr\Gamma_i(s)^H \succeq \gamma_i\matr I$.
\end{enumerate}
The same properties hold for the aggregated transfer function
$\matr{\tilde D}(s)$ with $\gamma=\min_i\gamma_i>0$.
\end{assumptions}

By inspection, these assumptions hold for the droop inverter $\matr{\tilde D}_\mathrm{droop}(s)$ given in~\eqref{eq:D-droop}.
Further, they are robust in the following sense:
Sufficiently small but arbitrary perturbations of $\matr{\tilde D}_\mathrm{droop}$ also satisfy them as long as the separation property continues to hold.
For an account of $k_i^q \le 0$, we refer to~\cite{niehuesSmallSignal2026}.

\subsection{Stability of the closed loop}
In this work, we give exact criteria for the stability of the closed-loop feedback interconnection depicted in Fig.~\ref{fig:system-full}, denoted by $\matr D(s)\#\matr N(s)$ in the following. The corresponding closed-loop transfer matrix reads
\begin{equation}
    \matr G(s) = \matr{\tilde D}(s)\left(\matr\eye+\matr{\tilde N}(s)\matr{\tilde D}(s) \right)^{-1}.
    \label{eq:closed-loop1}
\end{equation}
Properness of the interconnection is inherited from $\matr D_i(s)$ and $\matr N(s)$, such that we focus on its closed-loop modes in the following.
For a further discussion of the closed-loop representation and internal feedback stability, see~\cite{haberleDecentralizedParametricStability2026}

Before stating the main result, we exclude the trivial mode corresponding to a uniform shift of all phase angles $\theta_i$. This mode leaves all power flows unchanged and is therefore irrelevant for physical stability. Accordingly, $\matr\Xi$ always has the zero eigenvector
$\vec v_0=(\vec 1,\vec 0)^\top$, where $\vec 1\in\mathbb R^n$ denotes the vector of all ones. We define the orthogonal complement
$\mathcal D_\perp\coloneqq\{\vec x \in \mathbb{R}^{2n} \mid \vec v_0^\top\vec x=0\}$
and say that $\matr\Xi$ is positive definite on $\mathcal D_\perp$ if
$\vec x^H\matr\Xi\vec x>0$ for all nonzero $\vec x\in\mathcal D_\perp$.

By assumptions~\ref{ass:passivity}, $\matr{\tilde D}(s)$ is invertible, such that
\begin{equation}
    \matr G(s) = s\left(s\matr\Gamma(s)+\matr\Xi\right)^{-1} \eqqcolon s\matr\Phi(s)^{-1}. \label{eq:closed-loop2}
\end{equation}
For $s\neq0$, singularities of $\matr\Phi(s)$ correspond directly to poles of $\matr G(s)$.
The origin requires separate consideration because the prefactor $s$ may cancel poles of $\matr\Phi(s)^{-1}$.
In particular, the trivial phase-shift mode $\vec v_0$ is always located at $s=0$ and is irrelevant for physical stability, whereas any additional mode at the origin constitutes a hidden internal mode.
Hence, asymptotic stability modulo the trivial phase-shift mode requires that $\matr\Phi(s)$ be nonsingular throughout ${\overline{\mathbb C}_+\!}{\setminus}\{0\}$ and possess only the trivial mode at $s=0$.
We thus arrive at the following theorem.

\begin{theorem}
\label{thm:stability}
Under assumptions~\ref{ass:passivity}, the following statements are equivalent:
\begin{enumerate}
    \item The system $\matr D(s)\#\matr N(s)$ is asymptotically stable modulo the trivial phase-shift mode,
    \item $\matr\Phi(s)$ is nonsingular for all $s\in{\overline{\mathbb C}_+\!}{\setminus}\{0\}$ and $\ker\matr\Phi(0)=\operatorname{span}\{\vec v_0\}$,
    \item $\matr\Xi$ is positive definite on $\mathcal D_\perp$.
\end{enumerate}
\end{theorem}

\begin{proof}
By the discussion above, asymptotic stability modulo the trivial phase-shift mode is equivalent to $\matr\Phi(s)$ having no singularities in ${\overline{\mathbb C}_+\!}{\setminus}\{0\}$ and no additional zero modes at the origin.

We first consider $\matr\Xi$ nonsingular on $\mathcal D_\perp$.
Then the trivial zero eigenvalue associated with $\vec v_0$ is simple.
We prove that the number $\nu$ of negative eigenvalues of $\matr\Xi$ equals the number of singularities of $\matr\Phi(s)$ in $\mathbb C_+$, counted with multiplicity.

We prove this statement via homotopy, defining the family of matrices
\begin{align}
    \matr \Gamma_a(s) &\coloneqq (1-a) \matr \Gamma(s) + a  \matr \eye, \quad \matr \Phi_a(s) \coloneqq s \matr \Gamma_a(s) + \matr \Xi,
\end{align}
with $a \in [0,1]$. 
Since $\matr\Phi_a(0)=\matr\Xi$, the trivial zero mode is present for all $a$. As the zero eigenvalue of $\matr\Xi$ is simple and
$\Re \left( \vec v_0^H \matr\Gamma_a(0) \vec v_0 \right)>0$, 
the corresponding zero of $\det(\matr\Phi_a(s))$ at $s=0$ is simple.
We therefore define the deflated characteristic function $
f_a(s)\coloneqq s^{-1}\det\matr\Phi_a(s)$, which extends analytically to $s=0$ with $f_a(0)\neq0$.
Since $\matr\Gamma(s)$ is real-rational and has no poles in $\overline{\mathbb C}_+$, the functions $f_a(s)$ are analytic in $\mathbb C_+$ and continuous on $\overline{\mathbb C}_+$.

First consider the case $a=1$, for which $\matr\Phi_1(s)=s\matr\eye+\matr\Xi$. We see that $f_1(s)$ has $\nu$ zeros in $\mathbb C_+$ corresponding to the negative nontrivial eigenvalues of $\matr\Xi$.

Second, we show by contradiction that no zero can cross the imaginary axis as we vary $a$. Assume that $f_a(j\omega)=0$. Since $f_a(0)\neq0$, we have $\omega\neq0$, and there exists a normalized nonzero vector $\vec x$ such that $\matr\Phi_a(j\omega)\vec x=0$. Hence,
\begin{align}
    0 &= (1-a)j\omega \, \vec x^H\matr\Gamma(j\omega)\vec x + aj\omega +\vec x^H\matr\Xi\vec x. \label{eq:proof-quadratic}
\end{align}
By assumptions~\ref{ass:passivity}, $\Re \! \left(\vec x^H\matr\Gamma(j\omega)\vec x\right) \geq \frac{\gamma}{2}\|\vec x\|^2>0$, where here and below $\|\cdot\|$ denotes the Euclidean vector norm and its induced matrix norm. 
Since $\matr\Xi$ is real symmetric, $\vec x^H\matr\Xi\vec x$ is real. Taking the imaginary part of~\eqref{eq:proof-quadratic} therefore yields $0=\omega\left[
a+(1-a)\Re \! \left(
\vec x^H\matr\Gamma(j\omega)\vec x
\right)\right]$, which implies $\omega=0$, a contradiction.
Hence, no zero of $f_a$ can cross the imaginary axis.

Third, we show that no zero can escape to infinity as we vary $a$.
By assumptions~\ref{ass:passivity}, we have
$\matr \Gamma(s) + \matr \Gamma(s)^H \succeq \gamma \matr \eye $ for all $s\in\overline{\mathbb C}_+$.
Define $\gamma_0 \coloneqq \min\{\gamma/2,1\}$, which is independent of $a$.
Then we have for every $s\in\overline{\mathbb C}_+$ and every vector $\vec x$
\begin{equation}
    \|\matr \Gamma_a(s)\vec x\|\,\|\vec x\| \ge \left|\vec x^H\matr \Gamma_a(s)\vec x\right| \ge \Re\!\left(\vec x^H\matr \Gamma_a(s)\vec x\right) \ge \gamma_0 \|\vec x\|^2 .
\end{equation}
Hence, $\matr \Gamma_a(s)$ is uniformly invertible and $\| \matr \Gamma_a(s)^{-1}\| \le \gamma_0^{-1}$.
Now, we write $\matr\Phi_a(s)$ as
\begin{equation}
    \matr \Phi_a(s) = s\matr \Gamma_a(s)+\matr\Xi = s\matr \Gamma_a(s) \left(\matr \eye + \frac{1}{s}\matr \Gamma_a(s)^{-1}\matr\Xi \right).
\end{equation}
Choose $S_0 > \| \matr \Xi \| \gamma_0^{-1}$, independent of $a$.
For $|s| > S_0$ we have $\left\|\frac{1}{s}\matr \Gamma_a(s)^{-1}\matr\Xi\right\|<1$, such that the second factor is invertible by the Neumann series, and hence $\matr \Phi_a(s)$ is also invertible.
Consequently, $f_a(s) \neq 0$ for all $s\in\overline{\mathbb C}_+$ with $|s|> S_0$, uniformly for all $a \in [0,1]$. Therefore, no zero of $f_a$ can escape to infinity.

Finally, we apply the homotopy argument. For $a=1$, $f_1(s)$ has $\nu$ zeros in $\mathbb C_+$. As $a$ decreases from one to zero, no zero can cross the imaginary axis or escape to infinity.
Hence, the number of zeros in $\mathbb C_+$, counted with multiplicity, is invariant under the homotopy, and $f_0(s)$ also has exactly $\nu$ zeros in $\mathbb C_+$.

If $\matr\Xi$ is singular on $\mathcal D_\perp$, then $\matr\Phi(0)=\matr\Xi$ has at least one zero mode in addition to the trivial mode $\vec v_0$.
This mode may be hidden in $\matr G(s)=s\matr\Phi^{-1}(s)$ by cancellation with the factor $s$, but remains an internal mode of the closed-loop system.
Hence, the linearized closed-loop system is not asymptotically stable (but may be marginally stable).
In this case, linearization alone is inconclusive about the local stability of the original nonlinear system, which depends on higher-order terms.

The result follows. If $\matr\Xi$ is nonsingular on $\mathcal D_\perp$, the homotopy shows that $f_0$ has exactly $\nu$ zeros in $\mathbb C_+$, where $\nu$ is the number of negative eigenvalues of $\matr\Xi$.
If $\matr\Xi$ is singular on $\mathcal D_\perp$, then $\matr\Phi(0)=\matr\Xi$ has an additional nontrivial zero mode. 
Consequently, $\matr\Phi(s)$ is nonsingular in ${\overline{\mathbb C}_+\!}{\setminus}\{0\}$ and has only the trivial mode at the origin if and only if $\matr\Xi$ is positive definite on $\mathcal D_\perp$.
This proves the equivalence of all three statements.
\end{proof}

The matrix $\matr \Xi$ determines the stability of the inverter-based power grid.
In the following, we will introduce a graph-theoretic perspective on $\matr \Xi$ that allows us to reformulate the stability theorem and to derive simpler sufficient stability criteria.

\subsection{Tools from Algebraic Graph Theory}
Let $\mathcal{E}$ denote the set of edges in the effective network describing the coupling of the inverters and $\ell = |\mathcal{E}|$.
We label all connections either by a consecutive number $a \in \{1,\ldots,\ell\}$ or by their endpoints.
For every edge of the network, we fix an orientation to keep track of the direction of the power flow.
To keep track of the topology and the orientation, we define the signed node-edge incidence matrix $\matr E \in \mathbb{R}^{n \times \ell}$ with components $E_{i,a} = +1$ ($E_{i,a} = -1$) if branch $a$ starts (ends) at node $i$ and zero otherwise.
The unsigned incidence matrix $\matr F \in \mathbb{R}^{n \times \ell}$ has components $F_{i,a} = |E_{i,a}|$.
Furthermore, we have to quantify the meshing of the network.
To this end, we select $\varsigma = \ell - n + 1$ fundamental cycles of the network, i.e.~sets of oriented branches that form a closed loop.
We then introduce the cycle-edge incidence matrix $\matr C \in \mathbb{R}^{\ell \times \varsigma}$ with elements $C_{a,\sigma}=+1$ ($C_{a,\sigma}=-1$) if the (reversed) branch $a$ belongs to cycle $\sigma$ and zero otherwise.
The matrix $\matr C$ spans the space of cycle flows, i.e.~flows without sources or sinks.
This property is expressed by $\matr E \matr C = \matr 0$ and the orthogonality relation $\operatorname{im}(\matr E^\top) = \ker(\matr C^\top)$.

\subsection{Structure of the stability matrix \texorpdfstring{$\matr \Xi$}{Xi}}
According to theorem~\ref{thm:stability}, the stability of the grid is governed by $\matr \Xi$.
This matrix couples phase angles and voltage magnitudes, which complicates its analysis.
Nevertheless, it has a particular structure that enables the derivation of explicit stability criteria.

Using the tools introduced above, we can rewrite $\matr \Xi$ as
\begin{equation}
    \matr \Xi =\begin{pmatrix}
        \matr E \matr \QQ  \matr E^\top & \matr E \matr \PP  \matr F^\top \\
        \matr F \matr \PP \matr E^\top &  \matr R - \matr F \matr \QQ \matr F^\top
    \end{pmatrix}. \label{eq:Xi-block}
\end{equation}
The matrices $\matr \PP, \matr \QQ\in\mathbb{R}^{\ell\times\ell}$ are diagonal in edge space and summarize the real and reactive power exchanged through the grid.
In particular, the diagonal elements are given by
\begin{align}
    \QQ_{a,a} &= B_{\mathrm{from}(a),\mathrm{to}(a)} v_{\mathrm{from}(a)}^{\circ} v_{\mathrm{to}(a)}^{\circ} \cos\bigl(\theta_{\mathrm{from}(a)}^\circ - \theta_{\mathrm{to}(a)}^\circ \bigr), \\
    \PP_{a,a} &= B_{\mathrm{from}(a),\mathrm{to}(a)} v_{\mathrm{from}(a)}^{\circ} v_{\mathrm{to}(a)}^{\circ} \sin\bigl(\theta_{\mathrm{from}(a)}^\circ - \theta_{\mathrm{to}(a)}^\circ \bigr),
\end{align}
where $\mathrm{from}(a)$ and $\mathrm{to}(a)$ denote the terminal nodes of edge $a$.
Furthermore, $\matr R\in\mathbb{R}^{n\times n}$ is a diagonal matrix in the node space with the entries
\begin{equation}
    R_{ii} = \frac{1}{k_{i}^q} -2 B_{ii} v_i^\circ v_i^\circ > 0\, . \label{eq:Rii}
\end{equation}
We thus observe that $\matr \Xi$ can be written completely in terms of diagonal matrices and incidence matrices.
This allows for a reduction of the stability criteria in the spirit of the Schur complement.
In this way, we can find the following sufficient and necessary criterion for the positive definiteness of $\matr \Xi$ and thus for the small-signal stability of the network.

\begin{theorem}\label{thm:posdef-cycle}
Assume that $\cos(\theta_i^\circ - \theta_j^\circ) >0$ for all connections $(i,j)$.
Then the matrix $\matr \Xi$ is positive definite on $\mathcal{D}_{\perp}$ if and only if the matrix
\begin{align}
    \matr{\hat \Upsilon} &
    \coloneqq \matr{\Upsilon} + \matr{\Upsilon}_{\rm cycle},
    \quad \text{with}
    \\
    \matr{\Upsilon} &\coloneqq  \matr R - \matr F ( \matr  \QQ + \matr  \PP \matr \QQ^{-1} \matr  \PP ) \matr F^\top 
    \nonumber \\
    \matr{\Upsilon}_{\rm cycle} &\coloneqq \matr F \matr \PP \matr \QQ^{-1} \matr C (\matr C^\top \matr \QQ^{-1} \matr C)^{-1} \matr C^\top \matr \QQ^{-1} \matr \PP \matr F^\top,
    \nonumber 
\end{align}
is positive definite.
\end{theorem}

\begin{proof}
For arbitrary $\vec x, \vec z$ we have
\begin{equation}
    \begin{pmatrix} 
        \vec x \\ \vec z
    \end{pmatrix}^{\!\!\top}
    \matr \Xi
    \begin{pmatrix} 
        \vec x \\ \vec z
    \end{pmatrix} =
    \begin{pmatrix} 
        \vec y \\ \vec z
    \end{pmatrix}^{\!\!\top}
    \begin{pmatrix}
        \matr \QQ  &  \matr \PP  \matr F^\top \\
        \matr F \matr \PP & \matr R - \matr F \matr \QQ \matr F^\top
    \end{pmatrix}
    \begin{pmatrix} 
        \vec y \\ \vec z
    \end{pmatrix},
\end{equation}
where we have defined $\vec y \coloneqq \matr E^\top \vec x \in\mathbb{R}^\ell$.
Then $\vec y \in \operatorname{im}(\matr E^\top) = \ker(\matr C^\top)$, and conversely every $\vec y \in \ker( \matr C^\top)$ can be written as $\vec y = \matr E^\top \vec x$.
Hence $\matr \Xi \succ 0$ on $\mathcal{D}_\perp$ is equivalent to
\begin{equation}
    V(\vec y,\vec z) \coloneqq
    \begin{pmatrix} 
        \vec y \\ \vec z 
    \end{pmatrix}^{\!\!\top}
    \begin{pmatrix}
        \matr \QQ  & 
        \matr \PP  \matr F^\top \\
        \matr F \matr \PP & 
        \matr R - \matr F \matr \QQ \matr F^\top
    \end{pmatrix}
    \begin{pmatrix} 
        \vec y \\ \vec z 
    \end{pmatrix} > 0,
\end{equation}
for all $(\vec y,\vec z)\neq 0$ with $\matr C^\top \vec y = \vec 0$.
We expand the quadratic form as
\begin{equation}
    V(\vec y,\vec z) = \vec y^\top \matr \QQ \vec y + 2 \vec y^\top \matr \PP \matr F^\top \vec z + \vec z^\top (\matr R - \matr F \matr \QQ \matr F^\top) \vec z,
\end{equation}
for $\matr C^\top \vec y = \vec 0$.
To check whether the quadratic form is positive definite, we fix $\vec z$ and minimize $V(\vec y, \vec z)$ with respect to $\vec y$ under the constraint $\matr C^\top \vec y = 0$.
Since $\matr \QQ \succ 0$ by assumption, the problem is strictly convex and admits a unique minimizer.
Stationarity of the Lagrangian $\mathcal L(\vec y,\vec \lambda) = \vec y^\top \matr \QQ \vec y + 2 \vec y^\top \matr \PP \matr F^\top \vec z + \vec \lambda^\top \matr C^\top \vec y $
yields
\begin{align}
   \vec y = - \matr \QQ^{-1} \matr \PP \matr F^\top \vec z - \frac{1}{2} \matr \QQ^{-1} \matr C \vec \lambda.
\end{align}
Imposing $\matr C^\top \vec y = \vec 0$ and solving for $\vec \lambda$ yields
\begin{equation}
    \vec \lambda = -2 (\matr C^\top \matr \QQ^{-1} \matr C)^{-1} \matr C^\top \matr \QQ^{-1} \matr \PP \matr F^\top \vec z.
\end{equation}
Substituting back yields the minimizer
\begin{equation}
    \vec y^\ast( \vec z) = - \matr \QQ^{-1} \left( \matr \eye - \matr C (\matr C^\top \matr \QQ^{-1} \matr C)^{-1} \matr C^\top \matr \QQ^{-1} \right) \matr \PP \matr F^\top \vec z.
\end{equation}
Inserting $\vec y^\ast(\vec z)$ into the quadratic form $V(\vec y,\vec z)$ yields
\begin{equation}
    \min_{\matr C^\top \!\vec y = 0} V(\vec y, \vec z) = \vec z^\top \matr{\hat \Upsilon} \vec z,
\end{equation}
with $\matr{\hat \Upsilon}$ as stated above.
Since $\matr \QQ \succ 0$, we have $\vec y^\top \matr \QQ \vec y > 0$ for all nonzero feasible $\vec y$.
Therefore  $V(\vec y,\vec z) > 0$ for all feasible $ (\vec y,\vec z)\neq 0$ if and only if
\begin{equation}
    \min_{\matr C^\top\! \vec y = 0} V(\vec y,\vec z) > 0 \quad\text{for all } \vec z \neq 0.
\end{equation}
Using the previous computation, this is equivalent to $\vec z^\top \matr{\hat \Upsilon} \vec z > 0$ for all $\vec z\neq 0$, i.e.~$\matr{\hat \Upsilon}\succ0$.
\end{proof}

We remark that the cycle correction in theorem~\ref{thm:posdef-cycle} can be written in a more compact form.
Since $\matr \PP$ and $\matr \QQ$ are diagonal,
\begin{equation}
    \matr \PP \matr \QQ^{-1} = \operatorname{diag}\left(\frac{\matr \PP_{a}}{\matr \QQ_{a}}\right) = \operatorname{diag} \left(\tan(\theta_{\mathrm{from}(a)}^\circ-\theta_{\mathrm{to}(a)}^\circ) \right).
\end{equation}
We define $\matr Z = \matr F \matr \PP \matr \QQ^{-1} \matr C$ and $\matr \Lambda = \matr C^\top \matr \QQ^{-1} \matr C$.
If $\cos(\theta_i^\circ - \theta_j^\circ) >0$ for all connections $(i,j)$, then $\matr \Lambda$ is symmetric positive definite and the cycle correction can be written as
\begin{equation}
    \matr \Upsilon_{\rm cycle} = \matr Z \matr \Lambda^{-1} \matr Z^\top. \label{eq:cycle-simplified2}
\end{equation}

\subsection{Sufficient stability criteria}
The matrix $\matr{\hat \Upsilon}$ in theorem~\ref{thm:posdef-cycle} is rather complex as it includes all possible cycles of the network.
However, the cycle contribution is always positive semi-definite if $\cos(\theta_i^\circ - \theta_j^\circ) > 0$.
Hence, we can easily obtain sufficient stability criteria by omitting the term in~\eqref{eq:cycle-simplified2}.
We thus obtain the following result.

\begin{corollary}\label{cor:posdef-Upsilon}
If (i) $\cos(\theta_i^\circ - \theta_j^\circ) >0$ holds for all connections $(i,j)$ in the network and (ii) the matrix
\begin{equation}
    \matr \Upsilon = \matr R - \matr F  \big( \matr \QQ + \matr \PP^2 \matr \QQ^{-1} \big) \matr F^\top \label{eq:cycle-simplified}
\end{equation}
is positive definite, then the matrix $\matr \Xi$ is positive definite on $\mathcal{D}_\perp$, and the system is small-signal stable.
\end{corollary}

Before we proceed, we remark that the matrix $\matr \QQ + \matr \PP^2 \matr \QQ^{-1}$ is diagonal in edge space and has a rather peculiar form.
The diagonal entry associated with edge $(i,j)$ is given by
\begin{equation}
    W_{ij} = \big( \matr \QQ + \matr \PP^2 \matr \QQ^{-1} \big)_{(ij),(ij)} 
    = B_{ij}  \frac{v_i^\circ v_j^\circ}{\cos(\theta_i^\circ - \theta_j^\circ)}. 
    \label{def:matr-C} 
\end{equation}
It increases with the loading of the respective branch.

The matrix $\matr \Upsilon$ in~\eqref{eq:cycle-simplified} is defined in the space of nodes or buses.
We can derive a complementary stability criterion in edge space.
This allows us to assess how line loads can affect small-signal stability.

\begin{corollary}\label{cor:posdef-Upsilon-edge}
If the matrix
\begin{equation}
    \matr \Psi =  \big( \matr \QQ + \matr \PP^2 \matr \QQ^{-1} \big)^{-1} - \matr F^\top \matr R^{-1} \matr F
\end{equation}
is positive definite, then the matrix $\matr \Xi$ is positive definite on $\mathcal{D}_\perp$, and the system is small-signal stable.
\end{corollary}

\begin{proof}
Since $\matr R$ is positive definite, positive definiteness of $\matr \Psi$  implies the positive definiteness of the block matrix
\begin{equation}
    \begin{pmatrix}
        \matr R & \matr F \\
        \matr F^\top&
        \big( \matr \QQ + \matr \PP^2 \matr \QQ^{-1} \big)^{-1}
    \end{pmatrix},\label{eq:Upsilon-Psi-block}
\end{equation}
via the Schur complement.
Applying the Schur complement once more, now with respect to the lower right block, shows that both the matrix $\matr \Upsilon$ and the diagonal matrix $\matr \QQ + \matr \PP^2 \matr \QQ^{-1}$ are positive definite.
Using the components given in~\eqref{def:matr-C}, we find that $\cos(\theta_i^\circ - \theta_j^\circ) >0$ holds for all connections $(i,j)$ in the network.
Hence, stability follows from corollary~\ref{cor:posdef-Upsilon}.
\end{proof}

We close this section by remarking that the stability criteria in corollary~\ref{cor:posdef-Upsilon} and corollary~\ref{cor:posdef-Upsilon-edge} are formulated in terms of three rather simple matrices: the matrix $\matr R$ that is diagonal in the space of nodes or buses, the matrix $\matr \QQ + \matr \PP^2 \matr \QQ^{-1}$ that is diagonal in the space of edges, and the unsigned edge incidence matrix $\matr F$ connecting nodes and edges. 
The matrices $\matr \Upsilon$ and $\matr \Psi$ are essentially equivalent %for the stability analysis 
because they are both obtained via the Schur complement of the block matrix in~\eqref{eq:Upsilon-Psi-block}.

\subsection{A graph theoretical view}
The structure of the matrix $\matr \Upsilon$ allows for a graph-theoretic interpretation of the sufficient stability criteria derived above.
We rewrite the second term in $\matr \Upsilon$ in a form that is more common to algebraic graph theory~\cite{dorflerElectricalNetworksAlgebraic2018}.
We observe that the two expressions $- \matr F ( \matr \QQ + \matr \PP^2 \matr \QQ^{-1} ) \matr F^\top$ and $+ \matr E ( \matr \QQ + \matr \PP^2 \matr \QQ^{-1} ) \matr E^\top$ differ only in the sign of the diagonal elements.
We can absorb this difference in the diagonal part of $\matr \Upsilon$ by defining the diagonal matrix $\matr T$ with components
\begin{equation}
    T_{ii} = R_{ii} - 2\sum_{j\neq i}
    \frac{B_{ij} v_i^\circ v_j^\circ}{\cos(\theta_i^\circ - \theta_j^\circ)} = \frac{1}{k_i^q} - 2 \sum_{j=1}^n \frac{B_{ij} v_i^\circ v_j^\circ}{\cos(\theta_i^\circ - \theta_j^\circ)}.
\end{equation}
Then we can rewrite $\matr \Upsilon$ as
\begin{equation}
    \matr \Upsilon = \matr T + \matr E \big( \matr \QQ + \matr \PP^2 \matr \QQ^{-1} \big) \matr E^\top. \label{eq:Upsilon-T}
\end{equation}
By this transformation, the second term is a weighted graph Laplacian, which is positive semi-definite if all edge weights are positive.

To proceed, we define a signed weighted cone graph whose vertex set consists of all original nodes $i = 1,\ldots, n$ plus one apex node $i=0$ (Fig.~\ref{fig:Upsilon-Graph}).
All pairs of original nodes $(i,j)$ with $i,j>0$ are connected by edges with the weights $W_{ij}$ defined in Eq.~\eqref{def:matr-C}.

Furthermore, every node $i=1,\ldots,n$ is connected to the apex node $i=0$ by an edge with weight $W_{i0} = T_{ii}$.
Then $\matr \Upsilon$ is the grounded Laplacian of the cone graph.
If the Laplacian of the cone graph is positive semi-definite and has only the trivial zero mode, then the grounded Laplacian is positive definite.
Thus we can apply graph-theoretic methods to check whether $\matr \Upsilon$ is positive definite.
In particular, these methods remain applicable even when some of the edge weights $W_{i0}$ are negative~\cite{zelazo2014definiteness,songExtensionEffectiveResistance2019}.

\section{Local stability certificates}\label{sec:local}
In this section, we derive explicit local stability criteria from the general results derived in the previous section.
Local means that the criteria are formulated for individual nodes (edges) of the network and may include only properties of the respective node (edge) and the incident edges (nodes).
This is highly desirable for applications as it allows one to formulate regulations for individual machines independent of other parts of the grid.

\subsection{A local stability certificate for nodes}
We start with a nodal version that follows from corollary~\ref{cor:posdef-Upsilon}.

\begin{corollary}\label{cor:3}
If $\cos(\theta_i^\circ - \theta_j^\circ) > 0$ for all connections in the network and the effective droop gains satisfy
\begin{equation}
    \frac{1}{k_i^q v_i^\circ} > 2 \sum_{j=1}^n\frac{B_{ij} v_j^\circ}{\cos(\theta_i^\circ - \theta_j^\circ)}, \label{eq:stab-local-nodal}
\end{equation}
for all nodes $i = 1,\ldots, n$, then $\matr \Xi$ is positive definite on $\mathcal{D}_\perp$.
\end{corollary}

\begin{proof}
We start from~\eqref{eq:Upsilon-T}.
The criterion $\cos(\theta_i^\circ - \theta_j^\circ) > 0$ ensures that the second term $\matr E( \matr \QQ + \matr \PP^2 \matr \QQ^{-1})  \matr E^\top$ is positive semi-definite.
Criterion \eqref{eq:stab-local-nodal} ensures that all entries $T_{ii}$ of the matrix $\matr T$ are positive.
Hence, $\matr \Upsilon$ is positive definite.
\end{proof}

An equivalent sufficient criterion has previously been derived in~\cite{niehuesSmallSignal2026}.
The present approach illuminates how it deviates from necessary stability criteria and shows that it is generally not tight.
In comparison to the basic theorem \ref{thm:posdef-cycle}, we have made two approximations: (i) We have neglected the cycle components $\matr \Upsilon_{\rm cycle} = \matr{\hat \Upsilon}- \matr \Upsilon$. 
(ii) We have neglected the contribution from the Laplacian term $\matr E ( \matr \QQ + \matr \PP^2 \matr \QQ^{-1} ) \matr E^\top$, which may also be substantial, especially for high line loads.

\begin{figure}[t]
\centering
\begin{tikzpicture}[
  >=latex,
  nodebox/.style={draw,rectangle,align=center,inner sep=4.5pt,fill=white},
  physical edge/.style={draw=black,line width=0.7pt},
  augmented edge/.style={draw=green!60!black,line width=0.8pt},
  brace/.style={decorate,decoration={brace, amplitude=5pt},line width=0.7pt}
]
\node[green!60!black] at (6.5,2.6) {Cone graph};
\node[nodebox] (i0) at (2.2,2.1) {\(i=0\)};
\node[nodebox] (i1) at (0.8,0.7) {\(i=1\)};
\node[nodebox] (i2) at (2.4,0.6) {\(i=2\)};
\node[nodebox] (i3) at (1.2,-0.5) {\(i=3\)};
\node[nodebox] (i4) at (4.0,-0.5) {\(i=4\)};
\draw[augmented edge] (i0) -- (i1);
\draw[augmented edge] (i0) -- (i2);
\draw[augmented edge] (i0) -- (i3);
\draw[augmented edge] (i0) -- (i4);
\draw[physical edge] (i1) -- (i2);
\draw[physical edge] (i1) -- (i3);
\draw[physical edge] (i2) -- (i3);
\draw[physical edge] (i3) -- (i4);
\draw[physical edge] (i1.south east) to[bend right=0] (i4.west);
\draw[physical edge] (i2.south east) to[bend left=0] (i4.north west);
\draw[brace, green!60!black] (4.7,2.5) -- (4.7,1.1);
\draw[brace] (4.7,1.) -- (4.7,-0.8);
\node[align=flush left, color=green!60!black, text width=3cm] at (6.6,1.4) {Edges with weights $$W_{i0}{=}T_{ii}$$};
\node[align=flush left, text width=3cm] at (6.6,0.2) {Physical network $$W_{ij}{=}B_{ij}\frac{v_i^\circ v_j^\circ}{\cos\!\left(\theta_i^\circ - \theta_j^\circ\right)}$$};
\end{tikzpicture}
\vspace*{-0.2cm}\caption{
Graph-theoretic interpretation of corollary~\ref{cor:posdef-Upsilon}.
The matrix $\matr \Upsilon$ is the grounded Laplacian of the weighted cone graph shown here.
Small-signal stability is guaranteed whenever this grounded Laplacian is positive definite.}\label{fig:Upsilon-Graph}
\end{figure}
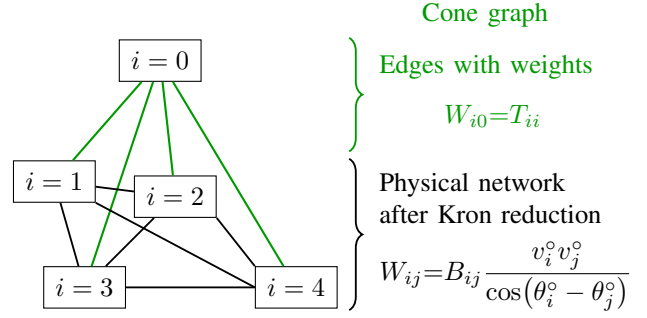

\subsection{Stability certificates for edge loads}
In the following, we derive local stability certificates for edges instead of nodes.
The resulting criteria can be interpreted in terms of the edge loads, which must not be too high.

The basic idea of our approach is to decompose the matrix $\matr \Upsilon$ into a sum over all edges in the network via suitable weight factors.
For every node $i$ we define a set of weights $c_{ij}>0$ for the adjacent edges such that $\sum_{j \neq i}c_{ij} = 1$.
Notably, we do not require $c_{ij} = c_{ji}$.
Then we obtain the following corollary

\begin{corollary}
If for all edges $(i,j)$, we have $\cos(\theta_i^\circ-\theta_j^\circ) > 0$ and the criterion
\begin{equation}
    \frac{\cos(\theta_i^\circ - \theta_j^\circ)}{B_{ij} v_i^\circ v_j^\circ}>  \frac{k_{i}^q/c_{ij}}{1 + 2 k_{i}^q  |B_{ii}| v_i^\circ v_i^\circ} + \frac{k_{j}^q/c_{ji}} {1 + 2 k_{j}^q |B_{jj}| v_j^\circ v_j^\circ} \label{eq:criterion-edgelocal-generalized}
\end{equation}
is satisfied, then the matrix $\matr \Xi$ is positive definite on $\mathcal{D}_\perp$, and the system is small-signal stable.
\label{cor:H}
\end{corollary}

\begin{proof}
Let $\vec e_i\in\mathbb{R}^n$ denote the $i$th standard basis vector and recall the edge weights $W_{ij}$ from~\eqref{def:matr-C} to write
\begin{align}
    \matr \Upsilon &= \matr R - \matr F \left( \matr \QQ + \matr \PP^2 \matr \QQ^{-1} \right)  \matr F^\top \nonumber\\
    &=\sum_{i=1}^n R_{ii} \vec e_i\vec e_i^\top - \sum_{i<j}(\vec e_i+\vec e_j)W_{ij}(\vec e_i+\vec e_j)^\top.
\end{align}
For any $c_{ij}>0$ with $\sum_{j\ne i}c_{ij}=1$, relabeling $i\leftrightarrow j$ when $i>j$ yields
\begin{equation}
    \sum_{i=1}^nR_{ii}\vec e_i\vec e_i^\top = \sum_{i<j}(R_{ii}c_{ij}\vec e_i\vec e_i^\top+R_{jj}c_{ji}\vec e_j\vec e_j^\top), 
\end{equation}
which, in turn, gives
\begin{equation}
    \matr \Upsilon = \sum_{i<j}
    \begin{pmatrix}
        \vec e_i & \!\!\!\vec e_j
    \end{pmatrix}
    \begin{pmatrix}
        R_{ii}c_{ij} - W_{ij} & - W_{ij} \\ - W_{ij} & R_{jj}c_{ji}- W_{ij}
    \end{pmatrix}
    \begin{pmatrix}
        \vec e_i^\top \\ \vec e_j^\top
    \end{pmatrix}.
\end{equation}
In this notation, criterion~\eqref{eq:criterion-edgelocal-generalized} reads $\frac1{W_{ij}} > \frac1{R_{ii}c_{ij}} + \frac1{R_{jj}c_{ji}}$ and directly implies $R_{ii}c_{ij}>W_{ij}$, $R_{jj}c_{ji}>W_{ij}$ and $\det=R_{ii}c_{ij}R_{jj}c_{ji}-W_{ij}(R_{ii}c_{ij}+R_{jj}c_{ji})>0$ for the determinant of the $2{\times}2$ matrix above, which implies positive definiteness by Sylvester's criterion, and after summation over $i<j$ positive definiteness of $\matr\Upsilon$.
Positive definiteness of $\matr \Xi$ on $\mathcal{D}_\perp$ and stability then follow from corollary~\ref{cor:posdef-Upsilon}.
\end{proof}

We remark that a special case of this criterion can also be obtained by applying diagonal dominance to the matrix $\matr \Psi$ in corollary~\ref{cor:posdef-Upsilon-edge}.
This yields the weight factors $c_{ij}=1/\deg(i)$ for all edges $(i,j)$ and $c_{ij}=0$ otherwise, where $\deg(i)$ denotes the degree of node $i$, i.e.~the number of neighbors of node $i$.

In the following, we will rather use the weight factors $c_{ij}=B_{ij}/\sum_{k \neq i} B_{ik}$ which splits the nodal contribution according to the strength of the edges.
Then, criterion~\eqref{eq:criterion-edgelocal-generalized} reads
\begin{equation}
    \frac{\cos(\theta_i^\circ - \theta_j^\circ)}{ v_i^\circ v_j^\circ} > \frac{k_{i}^q \sum_{k \neq i} B_{ik}}{1 + 2 k_{i}^q  |B_{ii}| v_i^\circ v_i^\circ} + \frac{k_{j}^q \sum_{k \neq j} B_{jk}} {1 + 2 k_{j}^q |B_{jj}| v_j^\circ v_j^\circ} \, .
    \label{eq:criterion-edgelocal-generalized-2}
\end{equation}
This criterion can guarantee stability if the grid load is not too high.
More precisely, phase angle differences $|\theta_i^\circ - \theta_j^\circ|$ must be small such that the cosine on the left-hand side is large.
Furthermore, the criterion is harder to satisfy if $v_i^\circ$ and $v_j^\circ$ differ strongly.

\section{Applications and Examples} \label{sec:examples}
Here, we demonstrate the applicability of the derived stability criteria on several examples and showcase the effect of the simplifications used to derive the local stability certificates.
We consider shuntless systems, where the susceptance matrix is treated as a Laplacian, with $B_{ii}=-\sum_{j\neq i}B_{ij}$.

\subsection{Two inverters}
We consider two droop-controlled inverters connected by one edge, $(1, 2)$.
Then, the unsigned incidence matrix is $\matr F = (1, 1)^\top$, and
\begin{equation}
    \Delta \theta_{12}^{\circ} = \theta_1^{\circ}-\theta_2^{\circ}, \quad
    W_{12} = \QQ_{12} + \PP_{12}^2 \QQ_{12}^{-1}
    = \frac{B_{12}v_1^{\circ}v_2^{\circ}}{\cos \left(\Delta \theta_{12}^{\circ}\right)} .
\end{equation}
Since the two-node network has no cycles, $ \hat{\matr \Upsilon} = \matr \Upsilon$.
Assuming $\cos \Delta \theta_{12}^\circ>0$, stability is determined by the matrix
\begin{equation}
    \matr \Upsilon = \matr R - W_{12} \matr F \matr F^\top =  \begin{pmatrix}
        R_{11} - W_{12} & - W_{12}\\
        -W_{12} & R_{22} - W_{12}
    \end{pmatrix},
\end{equation}
where $R_{ii}$ is given by~\eqref{eq:Rii}.
We then find that the linearized dynamics is asymptotically stable ($\matr \Upsilon \succ 0$) if and only if $\det(\matr \Upsilon) = R_{11} R_{22} - W_{12} (R_{11}+R_{22}) > 0$.
Substituting the full expressions, this criterion reads
\begin{align}
    &\left[\frac{1}{k_1^q}-2B_{11}\left(v_1^\circ\right)^2\right] \left[\frac{1}{k_2^q} - 2B_{22}\left(v_2^\circ\right)^2 \right] \\
    & \!\!\!> \frac{B_{12}v_1^\circ v_2^\circ}{\cos \left(\Delta\theta_{12}^{\circ}\right)} \left[ \frac{1}{k_1^q}-2B_{11}\left(v_1^\circ\right)^2 + \frac{1}{k_2^q}-2B_{22}\left(v_2^\circ\right)^2\right].\label{eq:2converters-stability1}
\end{align}
Let us now take $k_1^q = k_2^q = k^q$ and $v_1^\circ = v_2^\circ = v^\circ$.
Assuming that there are no shunts, we have $B_{11}=B_{22} = - B_{12}$.
Thus, we find the simplification
\begin{equation}
    k^q < \frac{1}{2 B_{12} (v^\circ)^2} \frac{\cos (\Delta \theta_{12}^\circ)}{1 - \cos (\Delta \theta_{12}^\circ)}. \label{eq:2inv-boundary}
\end{equation}
We show the border of the stable parameter region in Fig.~\ref{fig:Two-Inverters} for $B_{12} = 1$ and $v^\circ = 1$.
The reactive power droop gain $k^q$ must be small for a stable operation.
The upper limit depends crucially on the line load.
It is small for high loads but diverges in the limit of no load $\cos (\Delta \theta_{12}^\circ) \rightarrow 1$.

\begin{figure}
    \centering
    \includegraphics[width=\linewidth]{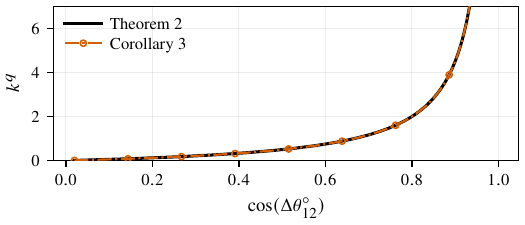}
    \vspace*{-0.6cm}\caption{
    Stability regions for two identical inverters:
    the upper limit for the effective reactive droop gain $k_1^q=k_2^q=k^q$ as a function of $\cos(\Delta\theta_{12}^{\circ})$.
    Small-signal stability is guaranteed below the curves.
    The exact result (theorem~\ref{thm:posdef-cycle}, black solid line) coincides with the sufficient local stability certificate (corollary~\ref{cor:3}).
    }\label{fig:Two-Inverters}
 \end{figure}

Under the assumption of homogeneous $k^q$ and $v^\circ$, \eqref{eq:2inv-boundary} coincides with the sufficient criterion of corollary~\ref{cor:3}.
The local stability certificate is both sufficient and necessary.

\subsection{Extended test cases} \label{sec:three-inverters}

\begin{figure*}[tb]
    \centering
    \includegraphics[width=\linewidth]{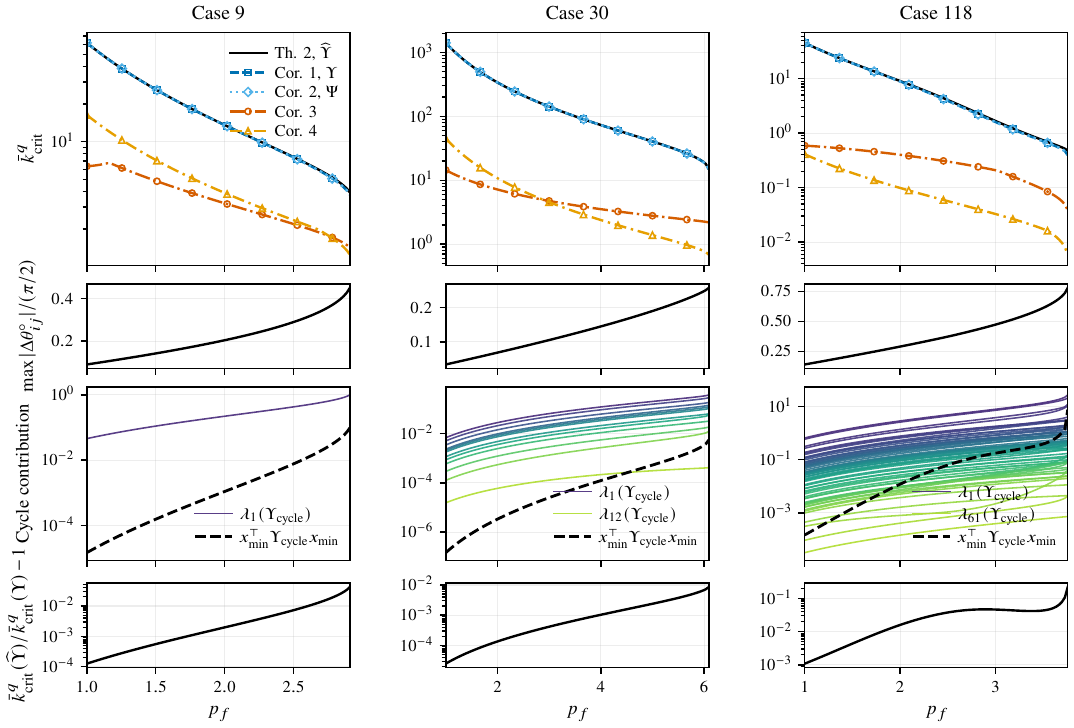}
    \vspace*{-0.6cm}
    \caption{
    Comparison of stability certificates and relevance of cycle contribution $\matr \Upsilon_{\rm cycle}=\hat{\matr \Upsilon} - \matr \Upsilon$ for lossless and shuntless IEEE Cases 9, 30, and 118.
    The loading factor $p_f$ is varied by scaling active power injections at all buses and reactive power setpoints at PQ buses.
    For each curve, values below the curve are certified stable by the corresponding criterion.
    The top row shows the critical values of the reactive power droop gain $\bar k^q$ for each IEEE case.
    At the plotted threshold, PV/slack buses have $k_i^q=\bar{k}^q_{\rm crit}$, while PQ buses have $k_i^q=0.1\bar{k}^q_{\rm crit}$.
    The second row shows the maximum angle difference over all edges in the sparse network.
    The third row shows the cycle contributions, i.e.~the spectrum of $\matr \Upsilon_{\rm cycle}$ and the projection $\vec x_{\min}^{\top} \matr \Upsilon_{\rm cycle} \vec x_{\min}$, where $\vec x_{\min}$ is the direction in which $\matr\Upsilon$ loses positive definiteness as the scalar droop gain $\bar k^q$ is increased.
    The bottom row shows the relative increase in the admissible value of $\bar{k}^q$ due to the cycle correction, comparing $\bar{k}^q_{\rm crit}(\hat{\matr \Upsilon})$ with $\bar{k}^q_{\rm crit}(\matr \Upsilon)$.
    Although $\matr \Upsilon_{\rm cycle}$ is not necessarily small, its projection onto the limiting mode is often much smaller than its leading eigenvalues, explaining why the cycle correction only weakly changes the certified stability boundary.
    Consequently, the critical values of the droop gain $\bar k^q$ in the exact criterion $\matr{\hat \Upsilon} \succ 0$ and the sufficient criterion $\matr \Upsilon \succ 0$ are very similar.}
    \label{fig:Criteria-Sparse}
\end{figure*}

In this section, we consider sparse transmission networks given by the IEEE test cases 9, 30 and 118.
Starting from the original test case operating point, we scale the active power injections $p_i^\circ$ at all buses, as well as the reactive power setpoints $q_i^\circ$ of the PQ buses, by a loading factor $p_f$. 
Voltage magnitudes at the PV buses and at the slack bus are kept fixed at their base values.
The remaining state variables, namely the non-slack voltage angles and the PQ-bus voltage magnitudes, are then computed from the lossless load-flow equations.
We increase $p_f$ as long as the load-flow equations converge.

For the small-signal stability analysis, we assume that generation at all PV buses and at the slack bus is provided by grid-forming inverters with homogeneous voltage-droop gain $k_i^q=\bar{k}^q$, such that we set $h_i=1$ in the profile $k_i^q= h_i \bar{k}^q$. 
At the remaining PQ buses, we assume smaller embedded inverter participation with weaker but non-negligible voltage droop, $h_i = 0.1$.

All certificates in Fig.~\ref{fig:Criteria-Sparse} (top rows) are evaluated for the same profile $k_i^q=\bar{k}^q h_i$, where $h_i$ is fixed by bus type, so that each curve gives an upper bound on the scalar droop level $\bar{k}^q$.
The exact boundary is obtained from the full matrix $\hat{\matr\Upsilon} = \matr \Upsilon + \matr \Upsilon_{\rm cycle} \succ 0$ in theorem~\ref{thm:posdef-cycle}.
Corollaries~\ref{cor:posdef-Upsilon} and~\ref{cor:posdef-Upsilon-edge} are equivalent.
They neglect the cycle contribution $\matr \Upsilon_\mathrm{cycle}$ in~\eqref{eq:cycle-simplified2} and yield the sufficient criterion $\matr \Upsilon \succ 0$.
In the current example, there is no discernible difference between the curves obtained from $\hat{\matr \Upsilon}\succ 0$ and $\matr \Upsilon \succ 0$.
We explain why the cycle contribution is negligible in the next section.

The local stability certificates derived in Sec.~\ref{sec:local} are more conservative.
We contrast corollary~\ref{cor:3} and corollary~\ref{cor:H} using the same droop profile $k_i^q=\bar{k}^q h_i$ as above, so that each certificate yields an upper bound on the scalar droop level $\bar{k}^q$.
For corollary~\ref{cor:H}, we choose weights proportional to the line susceptances, $c_{ij}=B_{ij}/\sum_{k\neq i}B_{ik}$.
As expected, the resulting bounds are below the global certificates.
The gap is especially pronounced in the larger test cases, where the local certificates discard more of the network structure.
The ordering of the two local certificates is not universal (cf.~corollaries~\ref{cor:3} and~\ref{cor:H} curves for Cases 9 and 30 in Fig.~\ref{fig:Criteria-Sparse}).
Corollary~\ref{cor:3} is limited by the worst node, while corollary~\ref{cor:H} is limited by the worst edge for the chosen weights.
As $p_f$ increases, the active limiting node or edge can change, producing kinks in the otherwise smooth curve (cf.~the corollary~\ref{cor:3} curves for Cases9 and 118 in Fig.~\ref{fig:Criteria-Sparse}).

\subsection{The cycle contribution}
We now explain why the cycle correction only weakly affects the certified stability boundary, even though $\matr \Upsilon_{\rm cycle}$ is not necessarily small.
Stability is determined by the lowest non-trivial eigenvalue of $\matr{\hat \Upsilon}$, which we denote as $\hat \mu_\mathrm{min}$ and assume to be simple.
To quantify the influence of the cycle correction, we regard $\matr\Upsilon_{\rm cycle}$ as a perturbation of $\matr{ \Upsilon}$.
Let $\mu_\mathrm{min}$ and $\vec x_\mathrm{min}$ denote the lowest eigenvalue of $\matr \Upsilon$ and its associated eigenvector.
First-order Rayleigh--Schrödinger perturbation theory yields~\cite{griffiths2018introduction}
\begin{equation}
    \hat \mu_\mathrm{min} = \mu_\mathrm{min} + \vec x_{\min}^{\top} \matr\Upsilon_{\rm cycle} \vec x_{\min} + \mathcal{O}\left( \matr\Upsilon_{\rm cycle}^2\right).
\end{equation}
Figure~\ref{fig:Criteria-Sparse} (third row) compares the first-order correction $\vec x_{\min}^{\top} \matr\Upsilon_{\rm cycle} \vec x_{\min}$ with the spectrum of $\matr\Upsilon_{\rm cycle}$.
For weak to modest grid loads, the correction is orders of magnitude smaller than the spectral scale of $\matr\Upsilon_{\rm cycle}$.
Hence, the small correction is not a consequence of $\matr\Upsilon_{\rm cycle}$ itself being small.
Instead, the stability-limiting mode $\vec x_{\min}$ is concentrated in the kernel and the eigenspaces corresponding to small eigenvalues of $\matr\Upsilon_{\rm cycle}$.

\section{Conclusion}
This paper establishes a necessary and sufficient criterion for the small-signal stability of lossless inverter-based power grids without assuming a restrictive device model.
By reformulating the stability problem in a graph-theoretic framework, we show how network topology, operating point, and inverter dynamics jointly determine stability through a single matrix criterion.
We identify increasingly decentralized stability certificates, clarifying the trade-off between exactness and locality while identifying the stabilizing network interactions neglected by local criteria.
Studies on IEEE benchmark systems demonstrate that the resulting certificates capture the stability boundary and quantify the conservatism introduced by decentralization.
Beyond providing new analytical tools for stability assessment, the proposed framework offers insights that may support the design of decentralized control strategies and future grid codes.

\section*{Acknowledgements}
JN gratefully acknowledges funding by Studienstiftung des Deutschen Volkes and by the Deutsche Forschungsgemeinschaft under Germany's Excellence Strategy -- The Berlin Mathematics Research Center MATH+ (EXC-2046/1, EXC-2046/2, No. 390685689), and by the OpPoDyn Project, Federal Ministry for Economic Affairs and Climate Action (No. 03EI1071A).
CD acknowledges funding through the Walter Benjamin program by the Deutsche Forschungsgemeinschaft (project No. 570237743).
LRG acknowledges funding through the NNLL project by NordForsk (No. 242039) and support through the Norwegian Research Council project FME SOLAR (No. 350244).
FH acknowledges funding by the German Federal Ministry for Economy and Energy (eKI4DS, No. 03EI1092A).

Generative AI was used for portions of the code and to help improve the narrative flow of this study.

\section*{Data Availability}
Data and source code supporting the findings in this study are openly available on Zenodo \cite{bacicCodeAccompanyingGraph2026}.

\bibliographystyle{ieeetr}
\bibliography{references.bib}

\end{document}